\documentclass[%
 reprint,
 amsmath,amssymb,
 aps,
]{revtex4-1}

\usepackage{graphicx}
\usepackage{dcolumn}
\usepackage{bm}
\usepackage{braket}
\usepackage{color}
\usepackage{latexsym}
\usepackage{amsmath,amsthm}

\newcommand{\modyseki}[1]{\textcolor[RGB]{151, 120, 241}{#1}}
\newtheorem{theorem}{Theorem}

\begin{document}

\preprint{APS/123-QED}

\title{
Improving the accuracy of the energy estimation by combining quantum annealing with
classical computation
}

\author{Takashi Imoto}
\author{Yuya Seki}%
\author{Yuichiro Matsuzaki}
 \email{matsuzaki.yuichiro@aist.go.jp}
\author{Shiro Kawabata}
 \email{s-kawabata@aist.go.jp}
\affiliation{%
 Research Center for Emerging Computing Technologies, National Institute of Advanced Industrial Science and Technology (AIST),
1-1-1 Umezono, Tsukuba, Ibaraki 305-8568, Japan.
}%

\date{\today}

\begin{abstract}
Quantum chemistry calculations are important applications of quantum annealing.
For practical applications in quantum chemistry,
it is essential to estimate a ground state energy of the Hamiltonian with chemical accuracy. 
However, there are no known methods to guarantee the accuracy of the estimation of the energy calculated by quantum annealing. 
Here, we propose a way to improve the accuracy of the estimate of the ground state energy
by combining quantum annealing with
classical computation.
 In our scheme,
 before running the QA, we need a pre-estimation of
 the energies of the ground state and first excited state with some error bars (corresponding to possible estimation error) by performing classical computation with some approximations. 
 We show that, if an expectation value and variance of the energy of the state after the QA 
 are smaller than 
 certain threshold values (that we can calculate from the pre-estimation), 
 the QA provides us with a better estimate of the ground state energy than that of the pre-estimation.
Since the expectation value and variance of the energy can be experimentally measurable by the QA, 
our results pave the way for accurate estimation of the ground state energy with the QA.

\end{abstract}

\maketitle


\section{\label{sec:level1}Introduction}

Recently, quantum chemistry calculations have attracted attention as new application for quantum devices
because of  
its potential use in medical areas. One of the main purposes of the quantum chemistry is to calculate the energy of the molecule Hamiltonian written by the second quantized form.
The high accuracy of the energy of chemical materials is required at least $1.6\times10^{-3}$ hartree where $1\mbox{hartee}=e^{2}/4\pi\epsilon_{0}a_{0}=27.211$eV and $a_{0}=1\mbox{bohr}=0.529\times10^{-10}\mbox{m}$.
This accuracy is called chemical accuracy.
The energy with a chemical accuracy
allows us to estimate
the chemical reaction rate at a room temperature 
using the Eyring equation\cite{eyring1935activated}

There are sophisticated techniques to map the molecule Hamiltonian 
with
the second quantized form into a spin Hamiltonian.
These techniques are important 
to implement the quantum chemistry calculations with the quantum devices composed of qubits, because the Hamiltonian to describe the molecules in the quantum devices should be written by the Pauli matrices.
We can
map the second-quantized many-body Hamiltonians onto those of qubits systems by Bravyi-Kitaev transformation\cite{bravyi2002fermionic, verstraete2005mapping, seeley2012bravyi, tranter2015b, xia2017electronic}.

There is
an improvement over the Jordan-Wigner transformation in the requirement for the number of the qubit operators per a fermionic operator.
The Jordan-Wigner transformation 
maps one of $n$ fermionic operators to $O(n)$ qubits operators.
On the other hand, Bravyi-Kitaev transformation maps one of  $n$ fermionic operators to $O(log(n))$ qubits.
The comparison of the gate number for the Bravyi-Kitaev tansformation and Jordan-Wigner transformation
to get the ground state and the lowest energy with the Trotter decomposition is reported\cite{tranter2018comparison}.
Also, Babbush \textit{et al}. represent the Hamiltonian using only 2-local interaction between spins\cite{babbush2014adiabatic}.

In fault tolerant quantum computation, quantum algorithms have been proposed in quantum chemistry calculations\cite{takeshita2020increasing, mueck2015quantum, babbush2018low}. 
Molecular energies are obtained using phase-estimation algorithms\cite{aspuru2005simulated, whitfield2011simulation}.
However, 
the fault tolerant quantum computer require
many 
qubits with high fidelity gate operations
beyond the capability of near-term quantum computer to use error-correction.
So the algorithm for quantum chemistry is not experimentally implemented with a practically useful size yet.

Recently, 
Noisy Intermediate-Scale Quantum(NISQ) computing is proposed \cite{arute2019quantum, zhang2020variational,endo2021hybrid}. 
One of the promising algorithms with NISQ is the variational quantum eigensolver(VQE) with the variational method\cite{peruzzo2014variational, mcclean2016theory}.
The VQE gives the lowest eigenvalue of a Hamiltonian such as that of a chemical material.
The VQE is a hybrid quantum-classical algorithm.
Variational algorithm is also used to simulate quantum dynamics
\cite{li2017efficient,chen2020demonstration}.
The energy variance was used to know how close the quantum state is to the energy eigenstate in the NISQ algorithm
\cite{kardashin2020certified}.

Quantum annealing(QA) is also a promising way to implement the quantum chemistry calculations. The QA was traditionally used 
to solve the combinatorial optimization problem\cite{kadowaki1998quantum, matsuzaki2020direct, seki2012quantum}. 
We map the combinatorial optimization problem 
into the Ising Hamiltonian $H_{P}$, and we call this a problem Hamiltonian whose ground state corresponds to the solution of the combinatorial optimization problem. 
On the other hand, we use another Hamiltonian $H_{D}$ that represents transverse magnetic fields, which we call a driver Hamiltonian.
In the QA, we prepare a ground state of $H_{D}$, and
the total time-dependent Hamiltonian is changed from $H_{D}$ to $H_{P}$ within an annealing time $T$.
As long as an adiabatic condition is satisfied, an adiabatic theorem guarantees that we can obtain the ground state of the problem Hamiltonian by the QA.
Importantly, by replacing the $H_P$ with the molecule Hamiltonian, 
the QA can be used to estimate an energy of the ground state in quantum chemistry\cite{xia2017electronic, copenhaver2020using, mazzola2017quantum, genin2019quantum, streif2019solving}.
 In addition, the excited states search in quantum chemistry is discussed \cite{teplukhin2019calculation, seki2020excited}.

D-wave systems, Inc\modyseki{.~}\cite{johnson2011quantum} have realized quantum annealing machines composed of thousands of the qubits. They use superconducting flux qubits to implement the quantum annealing.
There are many experimental demonstrations of the QA
by the device of the D-wave systems, Inc\modyseki{.~}\cite{kudo2018constrained, adachi2015application, hu2019quantum, kudo2020localization}. 
Especially, quantum chemistry calculations were demonstrated with the QA to estimate
the ground state energy for a small size molecule \cite{genin2019quantum}.
 
The potential problem to use the QA for practical quantum chemistry calculations is an intrinsic error during the QA. 
Non-adiabatic transitions induce a transition from the ground state to excited states. Also, decoherence due to the coupling with an environment
causes unwanted decay of the quantum states during the QA. 
Due to these problems, it is not clear whether we can achieve the chemical accuracy in quantum chemistry calculations by the QA. So it is essential to achieve a higher accuracy to estimate the ground state energy in the QA.


In this paper, we propose a way to estimate an energy of the target Hamiltonian with 
improved accuracy by combining quantum annealing with
classical computation.
We show that, if the population of the ground state is more than $1/2$ after the QA, an energy variance (that we can experimentally measure)
provides us with an upper bound of the estimation error. We also show a way to check whether the population of the ground state is more than $1/2$ after the QA or not by using classical computation. We need to know a possible range of the energies of the ground state and first excited state before the QA by performing classical computation with some approximation 
(such as a mean field technique).
We can calculate a certain threshold by using values from the pre-estimation, and if the energy estimated by the QA is smaller than the threshold, the population of the ground state is more than $1/2$ after the QA. 
Additionally, if the error bars (corresponding to the estimation error)
given by the pre-estimation is larger than an energy variance measured from the QA, 
we can use the energy variance as the improved error bars for the energy estimation.
The method are represented schematically in FIG\ref{fig:epsart}.

The paper is structured as follow.
In Sec. II, we review the QA.
In Sec III, we derive a relationship between the energy estimation error 
and the energy variance in the QA for the ground state search, and also show a condition when 
the energy variance becomes an upper bound of the energy estimation error.
In Sec IV, 
to check the performance of our scheme, we adopt our method to estimate a ground state energy of the hydrogen molecule.
In Sec V, we summarize and discuss our results.

\begin{figure}[ht]
\includegraphics[width=80mm]{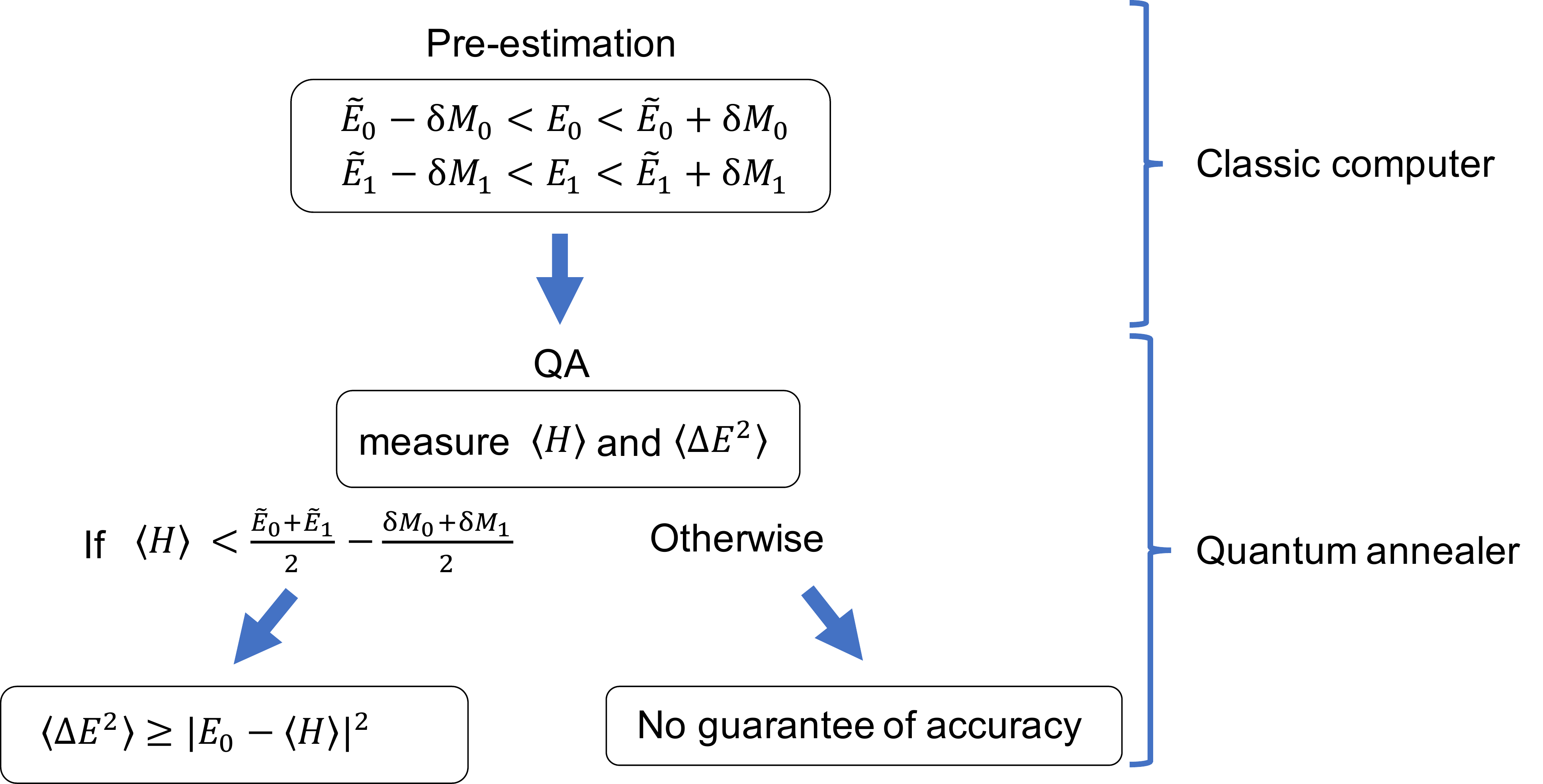}
\caption{A flow chart showing how to estimate the ground state energy of target Hamiltonian in our protocol.
We need to pre-estimate the ground state energy by using a classical computer with some approximation, and need to know $\tilde{E}_{0}$, $\tilde{E}_{1}$, $\delta M_0$, and $\delta M_{1}$ where
$\tilde{E}_{0}$ ($\tilde{E}_{1}$) is the approximated ground (first excited) state energy from pre-estimation and
$\delta M_0$ ($\delta M_{1}$)  is the bound of the error of the pre-estimation.
In addtion, $E_0$($E_1$) denotes the true energy of the ground (first excited) state and 
$\braket{H}$ ($\braket{\Delta E^{2}}$)
denotes the expectation (variance)
of the Hamiltonian of the state after the QA.
In our protocol, when $\braket{H}$ is smaller than $(\frac{\tilde{E}_0-\tilde{E}_0}{2}-\frac{\delta M_0 +\delta M_1}{2})$, the energy variance of the state after the QA can be an upper bound of the estimation error.
}

\label{fig:epsart}
\end{figure}

\section{Quantum annealing}

Let us review the QA for the ground state search.
We also regard the driving Hamiltonian as the transverse field.
The total Hamiltonian for the QA is described as follows
\begin{align}
    H(t)=\frac{t}{T}H_{P}+\Bigl(1-\frac{t}{T}\Bigr)H_{D}
\end{align}
where $T$ is the annealing time.
First, we prepare the ground state of the transverse field $H_{D}=-\sum_{i=1}^{N}\hat{\sigma}_{i}^{x}$, $\ket{\Psi(0)}=\ket{+\cdots+}$ where the quantum state $\ket{+}$ expresses the eigenstate of $\sigma^{x}$ with the eigenvalue $+1$.
Second, the driver Hamiltonian is adiabatically changed into 
the problem Hamiltonian.
Finally, we obtain the ground state of the problem Hamiltonian if the dynamics is adiabatic, and so the measurements of an observable $H_{P}$ provides the ground state energy.

Various noise deteriorates the accuracy of the QA. The main noise sources are environmental decoherence and non-adiabatic transitions. 
There is a trade-off between these two errors.
We should implement the QA slowly to avoid the non-adiabatic transitions, while the slower dynamics tend to increase the error due to the decoherence.

There are many attempts to suppress non-adiabatic transitions and decoherence.
The use of non-stochastic Hamiltonians has been proposed to increase an energy gap during the QA for a specific model, which could contribute the suppression of the non-adiabatic transitions. Inhomogeneous driving Hamiltonian for a p-spin model is known to
contribute to the speedup of the QA for specific cases
\cite{susa2018exponential, susa2018quantum}.
Direct estimation of the energy gap between the ground state and the first excited state using the quantum annealing was proposed, and this was shown to be more robust against non-adiabatic transitions than the conventional scheme \cite{matsuzaki2020direct}. 
Both theoretical and experimental studies have been made to suppress the decoherence during the QA.
We can use
error correction techniques \cite{pudenz2014error},
spin lock techniques \cite{chen2011experimental,nakahara2013lectures,matsuzaki2020quantum}, and decoherence free subspace 
\cite{matsuzaki2020quantum} for the suppresion of the decoherence. 
A method using non-adiabatic transition and quench for an efficient QA
is also investigated \cite{crosson2014different, goto2020excited, hormozi2017nonstoquastic, muthukrishnan2016tunneling, brady2017necessary, somma2012quantum, das2008colloquium}. 
Despite the efforts of previous research, there are no universal ways to suppress both environmental decoherence and non-adiabatic transitions during the QA, and this fact makes it difficult to guarantee the accuracy of the results of the QA.

\section{method}

Here, we present our scheme to estimate a ground state energy with 
improved accuracy in a certain condition.

\subsection{Bounds on the error of the energy}
\label{sec:bounds_on_the_error}

In practice, the quantum state become mixed states, because non-diagonal terms in the density matrix decay from the decoherence.
After implementing the QA, we measure only Hamiltonian and energy variance. 
In this case, we can show that the non-diagonal terms in the energy basis does not affect the expectation values. So we can describe the quantum state after the QA either a pure state or a mixed state as long as the energy population is the same between them. For simplicity, we would use a pure state for the description.


Suppose that we obtain a state of $\ket{\phi_{0}^{(\mbox{ann})}}$ after the QA.
We rewrite this state as follows.
\begin{align}
\ket{\phi_{0}^{(\mbox{ann})}}=\sqrt{1-\epsilon^{2}}\ket{\phi_{0}}+\sum_{m\neq 0}\epsilon_{m}\ket{\phi_{m}}.
\end{align}
where $\ket{\phi_{0}}$ denotes the ground state,  $\ket{\phi_{m}}$ 
$(m>0)$ denotes the $m$-th excited state, $\epsilon_{m}$ denotes the amplitude of the $m$-th excited state, and $\epsilon$ denotes the amplitude of all the states except the ground state.
In other words, $\sqrt{1-\epsilon^{2}}$ denotes the amplitude of the ground state. Due to the normalization, we have a condition of $\epsilon^{2}=\sum_{m\neq 0}\epsilon_{m}^{2}$.
Since we consider an expectation value of the Hamiltonian and the energy variance, the relative phase between the energy eigenstate does not affect our results. So
we can assume $\epsilon_{m}$ to be real values without loss of generality
throughout our paper. 

First, let us explain the estimation error and energy variance.
The estimation error of the energy eigenvalue of the problem Hamiltonian is given by

\begin{align}
    \bra{\phi_{0}^{(\mbox{ann})}}H_{P}\ket{\phi_{0}^{(\mbox{ann})}}-\bra{\phi_{0}}H_{P}\ket{\phi_{0}}=\sum_{m\neq 0}\epsilon_{m}^{2}(E_{m}-E_{0})
\end{align}
where $H_{P}$ is the problem Hamiltonian and $E_{m}$ is the $m$-th energy eigenvalue of the problem Hamiltonian.
On the other hand, the energy variance
$\Delta E^{2}$ is given by
\begin{align}
\Delta E^{2}&=\bra{\phi_{0}^{(\mbox{ann})}}H_{P}^{2}\ket{\phi_{0}^{(\mbox{ann})}} - \bra{\phi_{0}^{(\mbox{ann})}}H_{P}\ket{\phi_{0}^{(\mbox{ann})}}^{2}\notag\\
&=\sum_{m\neq 0}\epsilon_{m}^{2}(E_{m}-E_{0})^{2}-\Bigl(\sum_{m\neq 0}\epsilon_{m}^{2}(E_{m}-E_{0})\Bigr)^{2}.
\end{align}
We subtract the energy dispersion $\Delta E^{2}$ from the error squared of the energy as follows.

\begin{align}
    \Delta E^{2}&-\Bigl(\bra{\phi_{0}^{(\mbox{ann})}}H_{P}\ket{\phi_{0}^{(\mbox{ann})}}-\bra{\phi_{0}}H_{P}\ket{\phi_{0}}\Bigr)^{2}\notag\\
    &=\sum_{m\neq 0}\epsilon_{m}^{2}(E_{m}-E_{0})^{2}-\Bigl(\sum_{m\neq 0}\epsilon_{m}^{2}(E_{m}-E_{0})\Bigr)^{2}\label{eq:diff_disparsion_square_error}
\end{align}

We derive the next theorem. 
\begin{theorem}\label{theorem1}
If the amplitude of all the states except the ground state $\epsilon$ satisfies $\epsilon^{2}\leq\frac{1}{2}$, then 
\begin{align}
    \Delta E^{2}&-\Bigl(\bra{\phi_{0}^{(\mbox{ann})}}H_{P}\ket{\phi_{0}^{(\mbox{ann})}}-\bra{\phi_{0}}H_{P}\ket{\phi_{0}}\Bigr)^{2}\geq0.
\end{align}

\end{theorem}

\begin{proof}
we consider the relation between the energy variance and the estiamtion error of the energy eigenvalue. 
We remark that the following inequality is hold from the Cauchy–Schwarz inequality.
\begin{align}
    \Bigl(\sum_{m\neq n}\epsilon_{m}^{2}&(E_{m}-E_{n})\Bigr)^{2} \notag\\ 
    &\leq \Bigl(\sum_{m\neq n}\epsilon_{m}^{2}\Bigr)\Bigl(\sum_{m\neq n}\epsilon_{m}^{2}(E_{m}-E_{n})^{2}\Bigr).\label{eq:deform_cauchy_schwarz}
\end{align}
The lower bound of difference between the energy variance and the error of the square of energy (\ref{eq:diff_disparsion_square_error}) is given as follows.
\begin{align}
    \Delta E^{2}&-\Bigl(\bra{\phi_{0}^{(\mbox{ann})}}H_{P}\ket{\phi_{0}^{(\mbox{ann})}}-\bra{\phi_{0}}H_{P}\ket{\phi_{0}}\Bigr)^{2}\notag\\
    &=\sum_{m\neq 0}\epsilon_{m}^{2}(E_{m}-E_{0})^{2}-\Bigl(\sum_{m\neq 0}\epsilon_{m}^{2}(E_{m}-E_{0})\Bigr)^{2}\notag\\
    &\geq\sum_{m\neq n}\epsilon_{m}^{2}(E_{m}-E_{n})^{2}-2\epsilon^{2}\sum_{m\neq n}\epsilon_{m}^{2}(E_{m}-E_{n})^{2}\notag\\
    &=(1-2\epsilon^{2})\sum_{m\neq n}\epsilon_{m}^{2}(E_{m}-E_{n})^{2}
\end{align}
where we have applied the inequality (\ref{eq:deform_cauchy_schwarz}) to rewrite the inequality.
Finally, from $\sum_{m\neq n}\epsilon_{m}^{2}(E_{m}-E_{n})^{2}\geq0$ in the equations above, this completes the proof. 
\end{proof}

From the theorem\ref{theorem1}, we derive a condition for the energy variance to be an upper bound of the estimation error as follows.
\begin{align}
    \epsilon^{2}\leq\frac{1}{2} \label{eq:condition}.
\end{align}

\subsection{Pre-estimation of the energy before the QA by performing classical computation}\label{sec:check_pop}

In this subsection, we show a way to check $\frac{1}{2}\geq \epsilon^{2}$, i.e. 
the population of the ground state to be more than $1/2$. The main idea is to perform a pre-estimation of the energies of the ground state and first excited state. We could use a classical computer for such a pre-estimation by using a suitable approximation such as mean field technique or variational methods.
If this pre-estimation is accurate enough, the condition of $\frac{1}{2}\geq \epsilon^{2}$ is satisfied, and so we can use the energy variance to obtain an upper bound of the error estimation, as we will explain later.
Let $\tilde{E}_{n}$ denote the approximate value of obtained from the pre-estimation.

There are many ways to calculate the ground energy and excited state energy of molcules in quantum chemistry with a classical computer. For example, variational trial function gives us an upperbound of the ground state energy. There is a way to estimate the lower boud of the ground state energy as shown \cite{temple1928theory}. Also,various ways to obtain an energy gap between the ground state energy and exited states is known\cite{loos2020quest, nakatsuji2005deepening, cai2000low, silva2008benchmarks}. The combination of these technique would provide the range of the ground state energy and first excited state energy.

The estimation error of the approximate eigenvalue $\tilde{E}_{n}$ from the true eigenvalue $E_{n}$ is denoted by $\delta\tilde{E}_{n}$.
Thus, we remark the equality
\begin{align}
    \tilde{E}_{n}=E_{n}+\delta\tilde{E}_{n}\label{eq:tilde_energy}.
\end{align}
We assume that the estimation errors are bounded as follows.
\begin{align}
    |\delta\tilde{E}_{0}|<\delta M_{0},\ \   |\delta\tilde{E}_{1}|<\delta M_{1}.
\end{align}
where $\delta M_{0}$ and $\delta M_{1}$ denotes error bars
representing the accuracy of the pre-estimation.
We can show
that the sufficient condition of inequality $\frac{1}{2}\geq \epsilon^{2}$
is 
\begin{align}
    E_{0}^{(\mbox{ann})}\leq\frac{1}{2}(E_{0}+E_{1}).\label{eq:gs_condition}
\end{align}
Substituting the equality (\ref{eq:tilde_energy}) into the inequality (\ref{eq:gs_condition}), we obtain 
\begin{align}
    (E_{0}<)E_{0}^{(\mbox{ann})}<\frac{1}{2}(\tilde{E}_{0}+\tilde{E}_{1})-\frac{1}{2}(\delta\tilde{E}_{0}+\delta\tilde{E}_{1}) \label{eq:inequality_1}
\end{align}
We obtain a sufficient condition for the inequality (\ref{eq:inequality_1}) as follows.
\begin{align}
    (E_{0}<)E_{0}^{(\mbox{ann})}<\frac{1}{2}(\tilde{E}_{0}+\tilde{E}_{1})-\frac{1}{2}(|\delta\tilde{E}_{0}|+|\delta\tilde{E}_{1}|)\label{eq:inequality_2}
\end{align}
From $|\delta\tilde{E}_{0}|<\delta M_{0}$ and $|\delta\tilde{E}_{1}|<\delta M_{1}$, a sufficient condition for the inequality (\ref{eq:inequality_2}) is
\begin{align}
    (E_{0}<)E_{0}^{(\mbox{ann})}<\frac{1}{2}(\tilde{E}_{0}+\tilde{E}_{1})-\frac{1}{2}(\delta M_{0}+\delta M_{1}).\label{eq:condition2}
\end{align}
From the approximate energy eigenvalues by the pre-estiamtion ($\tilde{E}_{0}$ and $\tilde{E}_{1}$)
and the upper bound of the estimation errors 
($\delta M_{0}$ and $\delta M_{1}$), 
the inequality (\ref{eq:condition2}) is the sufficient condition of (\ref{eq:condition}).
This means that,
as long as
(\ref{eq:condition2}) is satisfied, we can use the energy variance as the new error bar (corresponding to the upper-bound of the estimation error) of the energy estimation.
Especially when the new error bar given by the energy variance is smaller than $\delta M_0$, the accuracy to estimate the ground state energy is better than that from just the pre-estimation.

The condition (\ref{eq:condition2})
is not always satisfied. If there are significant effect of decoherence and/or non-adiabatic transitions, $E_{0}^{(\mbox{ann})}$ could be large so that the sufficient conditions would not be satisfied. Alternatively, if the estimation error $(\delta M_{0}+\delta M_{1})$ is large, again, it becomes harder to satisfy the sufficient condition. 
In these cases, we should try other approaches such as optimizing the quantum annealing schedule, fabricating new samples with lower decoherence, or more precise pre-estimation with a longer calculation time using a classical computer to make the condition (\ref{eq:condition2}) satisfied.


\subsection{Measurement of the energy and variance of the Hamiltonian}
We describe how to measure the energy and variance of the Hamiltonian in the QA. We assume that we can perform any single qubit measurements in the QA. The Hamiltonian is now written in the form of $H=\sum _j \hat{P}_j$ where $\hat{P}_j$ denotes the product of Pauri matrices (such as $\hat{\sigma}^{z}_{0}$, $\hat{\sigma}^{z}_{1}$, $\dotsc$, and 
$\hat{\sigma}^{x}_{0}\hat{\sigma}^{x}_{1}\hat{\sigma}^{y}_{2}\hat{\sigma}^{y}_{3}$). After the preparation of the ground state with the QA, we can implement single qubit measurements to obtain $\langle \hat{P}_1 \rangle $. This means that, by repeating the experiments (that are composed of the ground state preparation and single qubit measurements), we can measure $\langle \hat{P}_j \rangle $ for every $j$, and we obtain $\langle H\rangle $ by summing up them. Similarly, we can measure $\langle H^2\rangle $, and so we can also measure the variance of the energy.
These techniques are used in the algorithm in NISQ devices\cite{cerezo2020variational, endo2020hybrid, cerezo2020variational, endo2020hybrid}.

\section{Numerical Result}
In this section, we perform the numerical simulations to estimate the error of the energy using our method.
We consider the hydrogen molecule.
The Hamiltonian of the hydrogen molecule can be described by the Pauli matrices.
To consider the decoherence,
we simulate the QA with the Lindblad master equation, and discuss
the relation between decoherence rate and the accuracy of the energy estimation. In addition, we plot the improved error bars obtained from our methods.

We introduce the Lindblad master equation.
We consider the time dependent system Hamiltonian $H(t)$ under a noisy environment.
The Lindblad master equation which we use in this paper is given by
\begin{align}
    \frac{d\rho(t)}{dt}=-i[H(t), \rho(t)]+\sum_{n}\gamma[\sigma^{(k)}_{n}\rho(t)\sigma^{(k)}_{n}-\rho(t)]
\end{align}
where $\sigma^{(k)}_{j}(k=x,y,z)$ denotes the Pauli matrix acting at site $j$, $\gamma$ denotes a decoherence rate and $\rho(t)$ is a density matrix of the quantum state at time $t$.
We solve the Lindblad master equation numerically with
the QuTiP~\cite{johansson184nation, johansson2012qutip}.
Throughout of this paper, we choose the decoherence type $\sigma^{z}_j$ as the Lindblad operator.
This type of noise has been studied in a previous work to consider the effect of noise on the superconducting qubits \cite{puri2017quantum}.

The Hamiltonian of hydrogen is given by 
\begin{align}
    H= 
    &h_{0}I+h_{1}\hat{\sigma}^{z}_{0}
    +h_{2}\hat{\sigma}^{z}_{1}
    +h_{3}\hat{\sigma}^{z}_{2}
    +h_{4}\hat{\sigma}^{z}_{3}\notag\\
    &+h_{5}\hat{\sigma}^{z}_{0}\hat{\sigma}^{z}_{1}
    +h_{6}\hat{\sigma}^{z}_{0}\hat{\sigma}^{z}_{2}
    +h_{7}\hat{\sigma}^{z}_{1}\hat{\sigma}^{z}_{2}
    +h_{8}\hat{\sigma}^{z}_{0}\hat{\sigma}^{z}_{3}
    +h_{9}\hat{\sigma}^{z}_{1}\hat{\sigma}^{z}_{3}\notag\\
    &+h_{10}\hat{\sigma}^{z}_{2}\hat{\sigma}^{z}_{3}
    +h_{11}\hat{\sigma}^{y}_{0}\hat{\sigma}^{y}_{1}\hat{\sigma}^{x}_{2}\hat{\sigma}^{x}_{3}
    +h_{12}\hat{\sigma}^{x}_{0}\hat{\sigma}^{y}_{1}\hat{\sigma}^{y}_{2}\hat{\sigma}^{x}_{3}\notag\\
    &+h_{13}\hat{\sigma}^{y}_{0}\hat{\sigma}^{x}_{1}\hat{\sigma}^{x}_{2}\hat{\sigma}^{y}_{3}
    +h_{14}\hat{\sigma}^{x}_{0}\hat{\sigma}^{x}_{1}\hat{\sigma}^{y}_{2}\hat{\sigma}^{y}_{3}\label{eq:hydrogen_hamiltonian}
\end{align}
where we used STO-3G basis and Jordan-Wigner transformation.
The coefficients of the Hamiltonian (\ref{eq:hydrogen_hamiltonian}) $h_{0},\dotsc ,h_{14}$ depend on the interatomic distance.
We consider the interatomic distance is $0.65$\AA. 
The coefficient in the Hamiltonian (\ref{eq:hydrogen_hamiltonian}) corresponding to the above interatomic distance is written in the TABLE \ref{tb:hydrogen_coefficient}, which is calculated by OpenFermion~\cite{mcardle2020quantum}.

\begin{table}[htb]
    \begin{tabular}{|c||c|}
    \hline
      $h_{0}$ & $0.03775110394645716$ \\
      $h_{1}$ & $0.18601648886230573$ \\
      $h_{2}$ & $0.18601648886230576$ \\
      $h_{3}$ & $-0.2694169314163209$ \\
      $h_{4}$ & $-0.2694169314163209$ \\
      $h_{5}$ & $0.1729761013074511$ \\
      $h_{6}$ & $0.0440796129025518$ \\
      $h_{7}$ & $-0.0440796129025518$ \\
      $h_{8}$ & $-0.0440796129025518$ \\
      $h_{9}$ & $0.0440796129025518$ \\
      $h_{10}$ & $0.1258413655800634$ \\
      $h_{11}$ & $0.1699209784826152$ \\
      $h_{12}$ & $0.1699209784826152$ \\
      $h_{13}$ & $0.1258413655800634$ \\
      $h_{14}$ & $0.17866777775953416$ \\ \hline
    \end{tabular}
  \caption{The coefficient of the hydrogen molecule. The unit of these values is GHZ, as we describe in the main text.}
  \label{tb:hydrogen_coefficient}
\end{table}

The most promissing device for the quantum annealing is a superconducting qubit. We mainly consider the implementation of the superconducting qubits. The typical energy scale of the superconducting qubit is aroung GHz \cite{clarke2008superconducting}. So we adopt thie energy scale to describe the Hamiltonian.

The relation between the measured energy and annealing time is shown in FIG\ref{fig:optimal_annealign_time}.
We can choose the annealing time to minimize the energy of the problem Hamiltonian after the QA. Throughout of our paper, we choose such an optimized annealing time for the plots.
Importantly, as the decoherence rate increases, the minimum energy after the optimization
increases. This is because the decoherence could induce a transition from a ground stat to excited states.

\begin{figure}[ht]
\includegraphics[width=80mm]{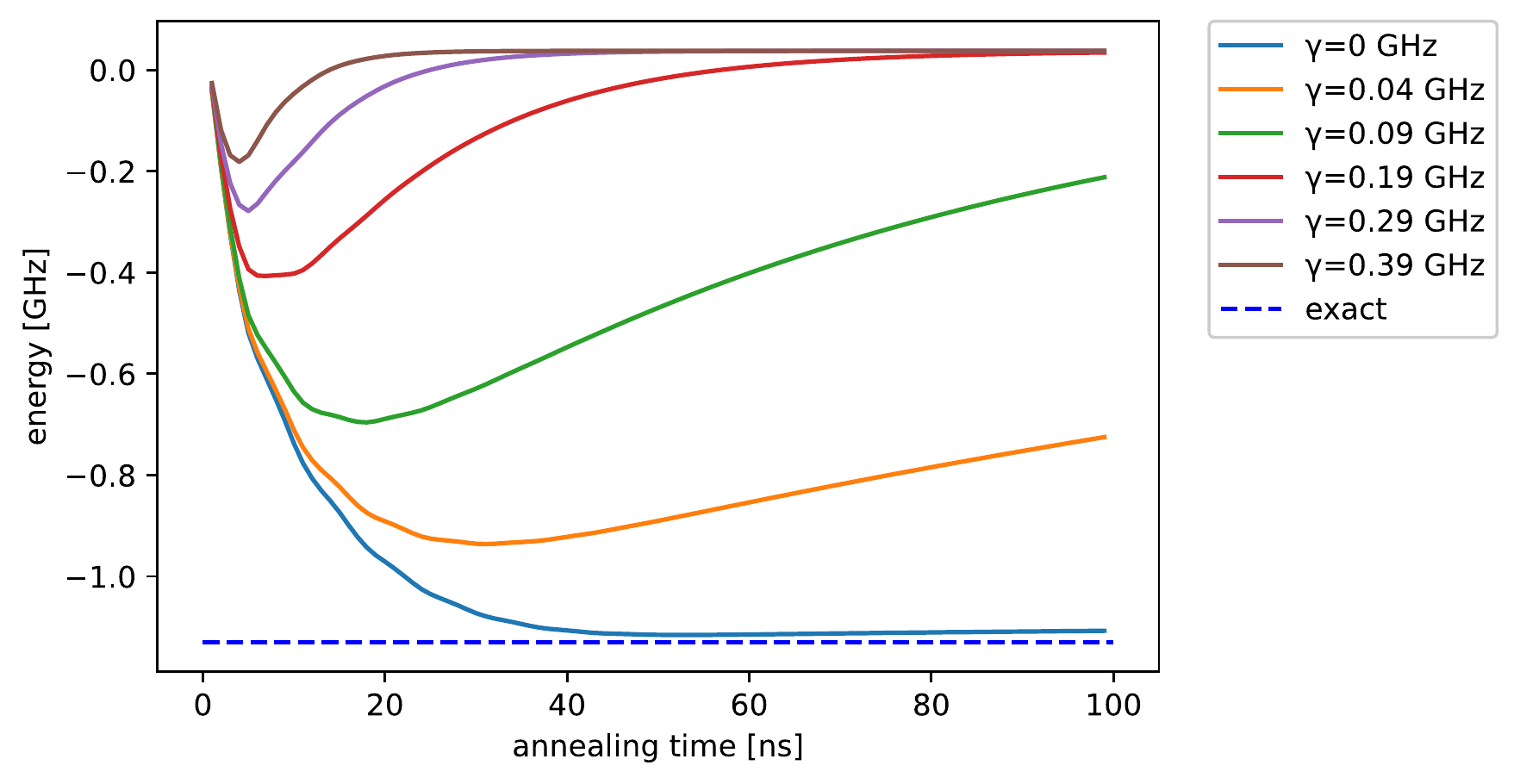}
\caption{
Plots to show the relation between the annealing time and the ground state energy by quantum annealing.
We consider the hydrogen molecule with interatomic distance $0.65$\AA.
Vertical axis denote the energy, while the horizontal axis denote annealing time $T$ for the each decoherence rate $\gamma$.
}
\label{fig:optimal_annealign_time}
\end{figure}

~~~


~~~


\begin{figure}[htbp]
\begin{center}
\begin{minipage}{0.48\textwidth}
\begin{center}
\begin{minipage}{0.90\textwidth}
\begin{center}
\includegraphics[width=0.95\textwidth]{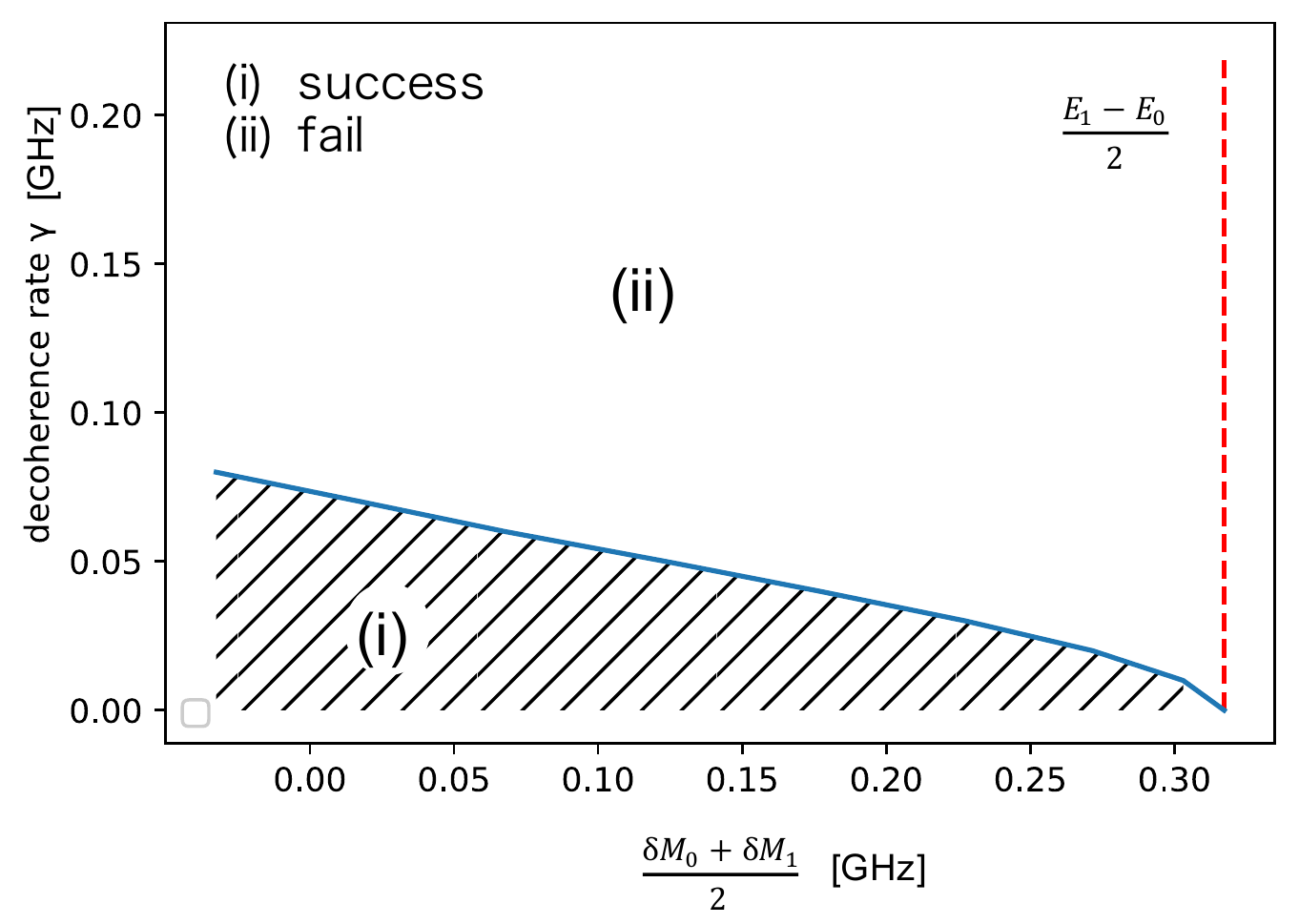}\\
(a) Threshold decoherence rate for our scheme to be applied
\end{center}
\end{minipage}\\
\begin{minipage}{0.90\textwidth}
\begin{center}
\includegraphics[width=0.95\textwidth]{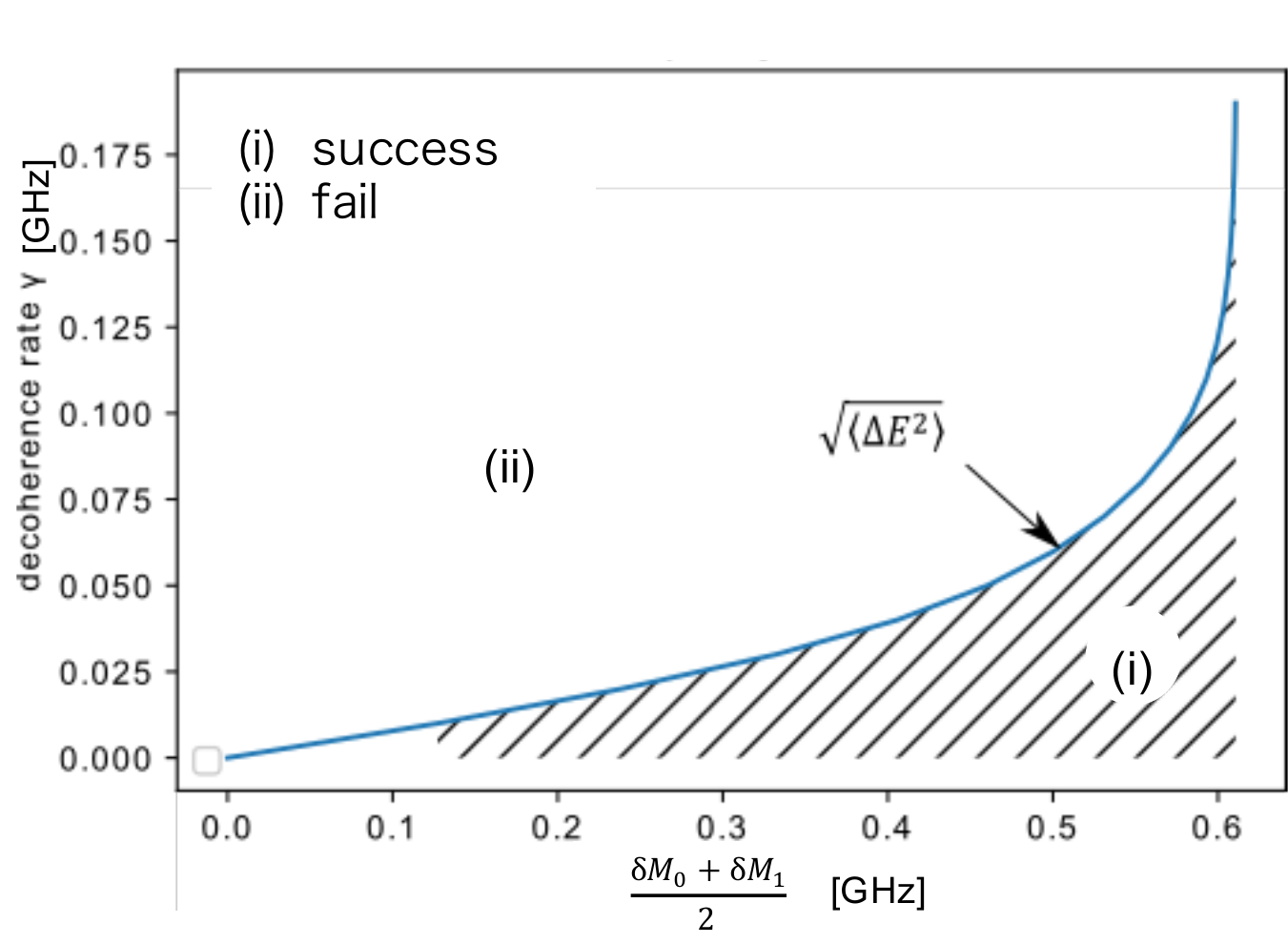}\\
(b) The threhold decoherence rate for our scheme to be more precise than pre-estimatation 
\end{center}
\end{minipage}
\end{center}
\end{minipage}
\caption{(a)The threshold decoherence rate for our scheme to be applied. The holizontal axis denotes the estimation error when we perform pre-estimation by a classical computer. As long as the decoherence rate of the QA is below the threshold, we can apply our scheme, and so the energy variance can be an upper-bound of the estimation error of the QA.
Here, we define that our scheme succeeds (fails) when we can (cannot) apply our scheme based on this prescription.
If $(\delta M_{0}+\delta M_{1})/2$ becomes equal to or larger than $(E_{1}-E_{0})/2$, our protocol always fails regardless the values of the decoherence rate.
(b)The threhold decoherence rate for our scheme to be more precise than pre-estimatation.
The holizontal axis denotes the estimation error when we perform pre-estimation by a classical computer.
As long as the decoherence rate of the QA is below the threshold, the energy variance is more precise upper-bound than pre-estimation error of the QA. Here, we define that our scheme succeeds (fails) when our scheme can (cannot) achieve better
estimation than the pre-estimation
on this prescription.%
}
\label{fig:accuracy_pre-estimate}
\end{center}
\end{figure}

By applying the method discussed in the subsection \ref{sec:check_pop}, we numerically determine the conditions that satisfies $\frac{1}{2}\geq\epsilon^{2}$ with the pre-estimation when we estimate the ground state energy of the hydrogen molecules. In other words, we show the region where we can use the energy variance as the upper bound of the estimation error. Such a region is plotted in FIG\ref{fig:accuracy_pre-estimate} (a). As the decoherence rate increases, pre-estimation should be done more precisely to satisfy $\frac{1}{2}\geq\epsilon^{2}$. On the other hand, even when we can use the energy variance obtained from the QA as the new error bound (due to the satisfaction of the condition of  $\frac{1}{2}\geq\epsilon^{2}$), the pre-estimation could be still better if the variance is too large. We plot the condition when the energy variance can be smaller than $\delta M_0$ while the condition of $\frac{1}{2}\geq\epsilon^{2}$ is satisfied, as shown FIG\ref{fig:accuracy_pre-estimate} (b).

In the FIG\ref{fig:accuracy_estimate}, we plot the estimation of the ground state energy and the error bar obtained by the energy variance when we use our scheme in the QA. As the decoherence rate becomes smaller, the error bar (corresponding to the energy variance) becomes smaller. Also, we confirm that the variance is actually larger than the estimation error when the decoherence rate is larger.

\begin{figure}[ht]
\includegraphics[width=80mm]{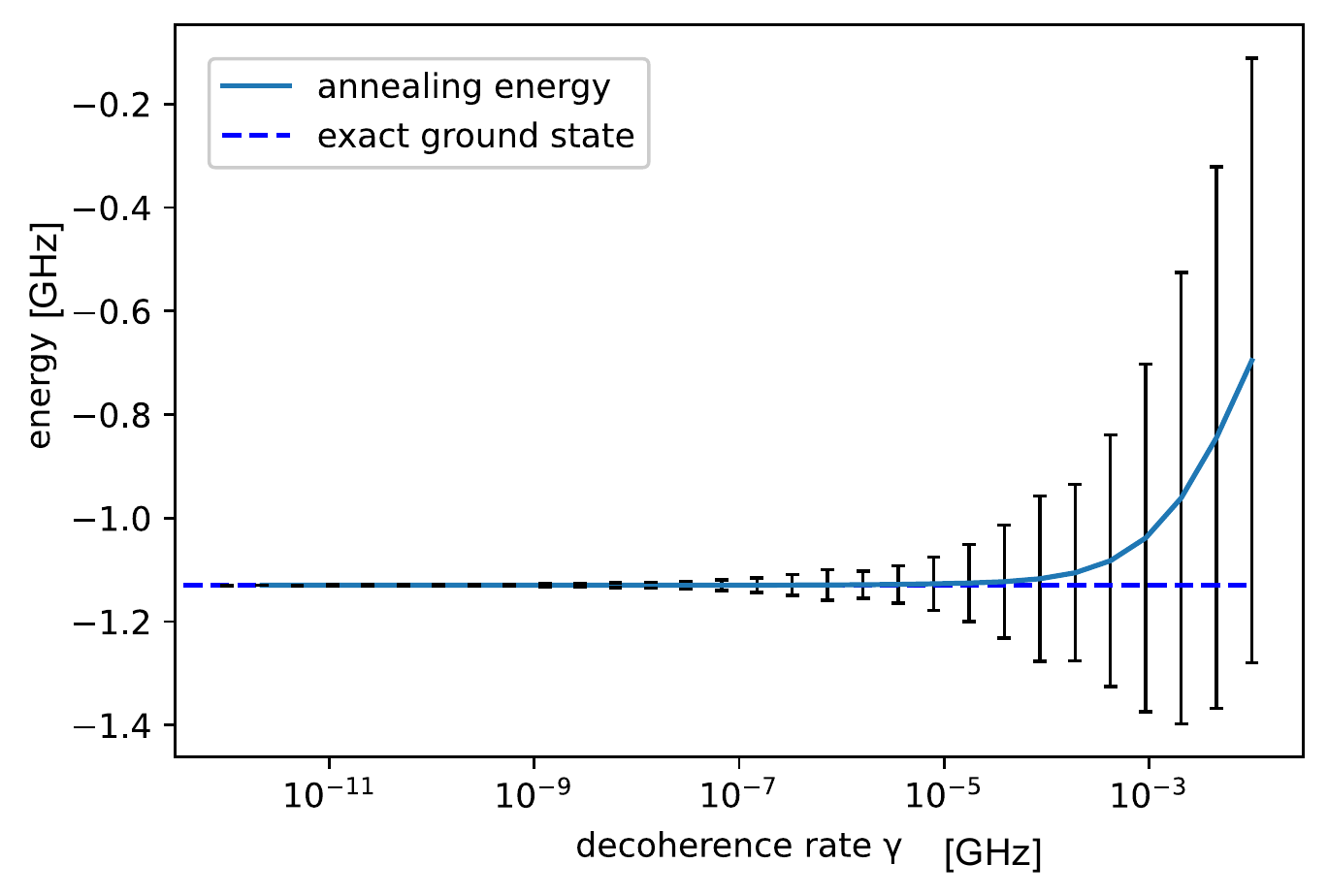}
\caption{
The energy expectation value with the error bar in our scheme.
The dashed line is the exact ground state energy.
The solid line is the energy expectation value ($E_{0}^{(ann)}$) obtained from the QA.
}
\label{fig:accuracy_estimate}
\end{figure}
\section{Conclusion}

In this paper, we propose a way to estimate an energy of the target Hamiltonian with improved accuracy by combining quantum annealing with classical computation. 
We show that,
if the population of the ground state is more than $1/2$ after the QA, the error of the energy for the problem Hamiltonian is upper bounded by the square root of the energy variance.
In order to check whether the population of the ground state is more than $1/2$ after the QA, we use classical computation for the pre-estimation of the energy of the ground state and first excited state. More precisely,
we obtain the approximate energy of the problem Hamiltonian with possible error bars for the ground state and first excited state by performing classical computation with some approximation (such as mean field technique).
From the values obtained by the pre-estimation, we can calculate a threshold, and if the energy of the state after the QA is smaller than the threshold, the population of the ground state is more than $1/2$ after the QA. In addition, if the energy variance of the QA is smaller than the error bar in the pre-estimation, we can use the energy variance as the improved error bar. Our methods are useful \textcolor{black}{to improve the accuracy}
for quantum chemistry calculations especially
when the QA with long-lived qubits is experimentally realized.

\begin{acknowledgments}
We thank a helpful discussion with Kenji Sugisaki.
This work was supported by Leading Initiative for
Excellent Young Researchers MEXT Japan and JST
presto (Grant No. JPMJPR1919) Japan. This paper is partly based on results obtained from a project,
JPNP16007, commissioned by the New Energy and Industrial Technology Development Organization (NEDO),
Japan. 
\end{acknowledgments}

~~~

\nocite{*}

\bibliography{apssamp}

\end{document}